\documentclass[12pt]{amsart} 
\usepackage{graphicx}
\usepackage{amssymb, amsmath, amsthm, mathtools}
\usepackage[margin=2in]{geometry} 

\usepackage{verbatim}

\usepackage{geometry}                		
\usepackage{graphicx}				
\usepackage{amssymb}

\usepackage{tikz}
\usepackage{alltt}
\usetikzlibrary{shapes}
\usetikzlibrary{plotmarks}

\newcommand{\dvol}{\mathrm{dvol}}
\newcommand{\e}{\varepsilon}

\newcommand{\norm}[1]{\left\lVert#1\right\rVert}

\newcommand{\cle}{\lesssim}

\newcommand{\R}{\mathbb{R}}
\newcommand{\geucl}{g_{\mathbb{E}}}

\newcommand{\rn}{d_{n}}

\newcommand{\dvolg}{\dvol_{\geucl}}

\newcommand{\Schw}{\mathrm{Schw}}

\renewcommand{\P}{\mathcal{P}}

\newcommand{\mE}{m_{\mathrm{eff}}}

\renewcommand{\H}{\mathcal{H}}

\newcommand{\ADM}{\mathrm{ADM}}

\newtheorem{theorem}{Theorem}
\newtheorem{definition}{Definition}
\newtheorem{remark}{Remark}
\newtheorem{lemma}{Lemma}

\renewcommand{\epsilon}{\varepsilon}
\renewcommand{\H}{\mathcal{H}}

\theoremstyle{plain}

\title{On Superposition of Relativistic point-sources}

\author{Noah Benjamin}

\author{Will McDermott}

\author{Iva Stavrov Allen}
\address{Lewis \& Clark College}
\email{wmcdermott@lclark.edu, istavrov@lclark.edu}

\begin{document}

\begin{abstract} The non-linearity of general relativity makes it at least difficult if not impossible to view a relativistic cloud of matter as being made up of point-source constituents. Perhaps the most delicate issue to circumnavigate is the inherent lack of the classical notion of superposition. Even if one were to believe that the recent framework developed by the first and the third author in \cite{NI} leads to an appropriate interpretation of the phrase ``initial data for a point-source", there is prima facie no reason to believe that it lends itself to a principle of superposition. In this paper we propose an extension of said framework which serves as a de-facto superposition of point-sources and which recovers Brill-Lindquist metrics in the limit. We also show that our proposal can be seen as a continuous extension of the classical superposition principle of Newtonian gravity. This paper fits within a larger program of representing relativistic clouds of matter as cumulative effects of point-sources.
\end{abstract}

\maketitle

\section{Introduction}
\subsection{Relativistic Poisson problem}\label{intro-RPP}
In this paper we are concerned with asymptotically Euclidean, conformally flat, time-symmetric initial data for  compactly supported, smooth clouds of relativistic dust on $\R^3$. Within such a framework the Einstein constraint equations reduce to the Hamiltonian constraint relating scalar curvature to the matter content. This constraint is typically formulated in terms of matter density of the dust cloud. In place of matter density we express matter (dust) by means of a $3$-form $\omega$. Our decision to use forms in place of densities is motivated by the fact that the concept of matter density inherently involves the concepts of metric and volume. As a result, employment of matter density in the formulation of the constraints makes it impossible to view the constraints as equations which inform us about the geometric responses to presence of matter. By expressing the ``amount" of matter present in a metric \emph{independent} way, such as the one in Definition \ref{matterdist:defn} below, we are able to frame the solutions to the constraints as responses to presence of matter. For additional benefits of using $3$-forms for conveying matter see Remark \ref{oldformulationisbad} and Section \ref{heuristics-subsection}. 

\begin{definition}[Matter Distribution]\label{matterdist:defn}
By matter distribution we mean a $3$-form $\omega = \phi\, \dvolg$ on $\R^3$ where $\phi \geq$ 0 is smooth, compactly supported, and not identically equal to zero.
\end{definition}

The Hamiltonian constraint for the metric $g_\omega$ corresponding to the matter distribution $\omega$ reads as 
$$R(g_\omega)\dvol_{g_\omega}=32\pi \frac{G}{2c^2} \omega.$$
From now on we set $g_\omega=\theta^4 \geucl$. The asymptotic conditions which ensure asymptotically Euclidean data are
$$\left|\partial_x^l\!\left(\theta(x)-1\right)\right|=O(|x|^{-|l|-1}),\ \ |x|\to \infty.$$
We see from $R(\theta^4\geucl)=-8\theta^{-5}\Delta_{\geucl}\theta$ that  
the Hamiltonian constraint is equivalent to 
\begin{equation}\label{TheEqn}
\theta \Delta_{\geucl}\theta \dvol_{\geucl}=-4\pi \tfrac{G}{2c^2} \omega \text{\ \ i.e\ \ } \theta \Delta_{\geucl}\theta =-4\pi \tfrac{G}{2c^2}\phi.
\end{equation}
It is interesting to note that the collection of papers by Arnowitt, Deser and Misner (e.g \cite{ZeroMass, ADM}) also featured the Hamiltonian constraint in this form.

\begin{definition}[RPP]\label{RPP:defn}
A Relativistic Poisson Problem (RPP) is the equation 
\begin{equation*}
\theta \Delta_{\geucl} \theta\, \dvol_{\geucl}= -4\pi \tfrac{G}{2c^2}  \omega
\end{equation*}
paired with an asymptotic boundary condition
\begin{equation*}
\lim_{|x| \to \infty} \theta (x) = b \ \ \text{where} \ b \geq 0.
\end{equation*}
When the asymptotic boundary condition is not explicitly mentioned the reader should assume $b=1$.
\end{definition}

The existence and uniqueness of solutions of RPP is addressed in Proposition 1, Proposition 8, and Remark 1 of \cite{NI}. Overall, we have the following result. 

\begin{theorem}\label{ThetaF:prop}
Let $\omega$ be a matter distribution. For each $b \geq 0 $, there exists a unique, smooth, and positive solution $\theta$ of
\begin{equation*}
\theta \Delta_{\geucl}\, \theta\dvolg  = -4\pi \tfrac{G}{2c^2} \omega , \  \lim_{|x| \to \infty} \theta(x) = b.
\end{equation*}
Furthermore, the solution $\theta$ satisfies the asymptotic conditions \begin{equation}\label{asymptotic-conditions}
\left|\partial ^l_x\left(\theta(x)-b\right)\right|\cle |x|^{-1-|l|} \text{\ \ as\ \ } |x|\to \infty.
\end{equation}
\end{theorem}

\begin{remark}\label{oldformulationisbad}
Had we formulated the Hamiltonian constraint as $R(g_\omega)=16\pi \frac{G}{c^2} \phi$ with $\phi$ denoting the energy \underline{density}, we would have been lead to the problem 
\begin{equation}\label{oldRPP}
\theta^{-5} \Delta_{\geucl} \theta  = -4\pi \tfrac{G}{2c^2} \phi , \  \lim_{|x| \to \infty} \theta(x) =1.
\end{equation}
The exponent on $\theta$ makes all the difference: in contrast to Theorem \ref{ThetaF:prop} the problem \eqref{oldRPP} does not permit positive solutions when $\phi$ is relatively large. Specifically, for $\phi$ which satisfy 
$$\frac{G}{2c^2}\int_\xi \frac{\phi(\xi)}{|x-\xi|}\dvolg \ge 1 \text{\ \ for all\ \ } x\in \mathrm{supp}(\phi)$$
one can use an iterative, inductive argument based on Green's Representation Formula 
$$\theta(x)=1+\frac{G}{2c^2} \int_\xi \frac{\phi(\xi)}{|x-\xi|}\theta^5(\xi)\dvolg$$
to show that any solution $\theta$ to \eqref{oldRPP} must satisfy $\theta\ge 2^{5^n}$ for all $n$, at least over the support of $\phi$.
\end{remark}

\subsection{Physical interpretations}\label{heuristics-subsection}
There are at least two additional good reasons for viewing the equation \eqref{TheEqn} as a relativistic counterpart to the Poisson equation of classical, Newtonian gravity. First and foremost, note that the Ansatz
$$\theta=1+\tfrac{G}{2c^2}u$$
converts \eqref{TheEqn} to
$$\Delta_{\geucl}u+ O\left(\frac{1}{2c^2}\right)=-4\pi G \phi;$$
as a result, the Newtonian limit of the equation \eqref{TheEqn} is the classical Poisson equation for the gravitational potential $u$. Secondly, there is a heuristic argument in favor of \eqref{TheEqn} based on the equivalence of mass and energy. A system of sources of \emph{effective} masses $m_i$ at locations $p_i$ would have the total mass-energy of 
$$\sum_i m_i c^2+\sum_{i\neq j} G\frac{m_im_j}{|p_i-p_j|}.$$
At first glance, it may appear that the latter could be reformulated as the statement that the (total) \emph{bare} mass of the system is given by  
\begin{equation}\label{CheatingEqn}
\begin{aligned}
\omega_{\mathrm{bare}}=&\sum_i m_i\delta_{p_i}(x) + \frac{G}{2c^2} \sum_{ij}  \frac{m_i m_j}{|x-p_j|}\delta_{p_i}(x)\\
=&\left(1+\frac{G}{2c^2}\sum_j\frac{m_j}{|x-p_j|}\right)\sum_i m_i\delta_{p_i}(x),
\end{aligned}
\end{equation}
where $\delta_{p_i}$ denotes the Dirac delta distribution with center at $p_i$. Substituting 
$$\theta=1+\tfrac{1}{2c^2}u \text{\ \ for\ \ } u(x)=\sum_jG\frac{m_j}{|x-p_j|}$$
into \eqref{CheatingEqn}, and observing that 
$$\Delta_{\geucl}\theta=-4\pi \tfrac{G}{2c^2} \sum_i m_i\delta_{p_i}$$
yields \eqref{TheEqn}. As we discuss below (see, for example, the self-interaction term in \eqref{meet-delta-over-r:eqn}) there is an illuminating error in \eqref{CheatingEqn}.  Nonetheless, the perspective presented in \eqref{CheatingEqn} makes clear the following:
\begin{itemize}
\item The solutions $\theta=1+\tfrac{1}{2c^2}u$ of \eqref{TheEqn} serve as generalizations of the gravitational potential from classical gravity.
\medbreak
\item The equation \eqref{TheEqn} de-facto decomposes bare mass into the sum of effective mass (corresponding to $\Delta_{\geucl} \theta$) and interaction energy (corresponding to $(\theta-1)\Delta_{\geucl}\theta$).
\medbreak
\item The coupling of $\theta$ and $\Delta_{\geucl}\theta$ in \eqref{TheEqn} is modeling gravitational interaction. 
\end{itemize}
The last observation can serve as a basis for explorations of one-parameter families of theories which continuously interpolate between Newtonian and relativistic dust; see Section \ref{furthercontext:sec} below.

\subsection{Brill-Lindquist metrics}\label{intro-BL}
The following example offers some crucial insights, although nominally it does not fit the conditions of our Definition \ref{matterdist:defn}. The example is a vacuum example, is defined on $\mathbb{R}^3\smallsetminus\{p_1, ..., p_Q\}$ and features the metric
\begin{equation}\label{BL-maineqn}
g_{BL}=\left(1+\frac{G}{2c^2}\sum_{i=1}^Q \frac{a_i}{|x-p_i|}\right)^4 \geucl,\ \ a_i>0.
\end{equation}
Metrics of the form \eqref{BL-maineqn} were studied in great detail in the work of Brill and Lindquist. Specifically, in \cite{BL} it is argued that metrics \eqref{BL-maineqn} describe a cloud of particles of (effective) masses $a_i$ located at $p_i$. 

Under certain assumptions on the separations $|p_i-p_j|$ and coefficients $a_i$ (for example, see \cite{SormaniStavrov}) the geometry of Brill-Lindquist metrics is as indicated in Figure \ref{fig1}. The front row of point-sources in Figure \ref{fig1} illustrates the fact that the minimal surfaces associated with individual particles, as well as the lengths of ``individual necks", depend on the mutual relationship between the values of $a_i$ and $p_i$.  

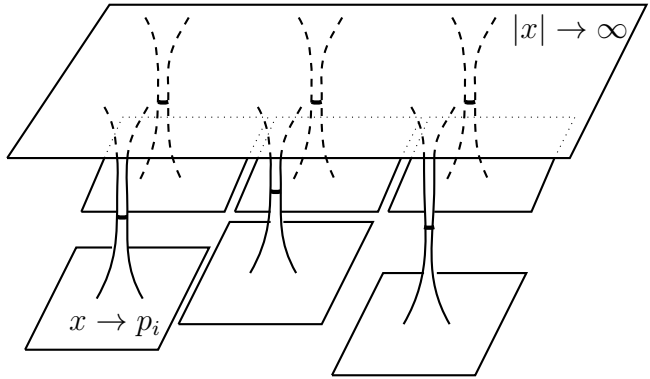
\begin{figure}[h]
\centering
\begin{tikzpicture}[scale=.68]
\draw[thick] (-8.5,-6) to (2.5, -6) to (4,-3) to (-6.5, -3) to (-8.5, -6);
\node[right] at (1.15,-3.5) {$|x|\to \infty$};

\draw[thick] (-8.15,-9.75) to (-5.55, -9.75) to (-4.55,-7.75) to (-6.15, -7.75); 
\draw[thick] (-6.45, -7.75) to (-7.15, -7.75) to (-8.15, -9.75);
\draw[thick, dashed] (-6.6, -5) to [out=-60, in=95] (-6.35, -6);
\draw[thick] (-6.35, -6) to [out=-85, in=90] (-6.35, -6.75) to [out=-90, in=60](-6.75, -8.75);
\draw[thick, dashed] (-5.75, -5) to [out=-120, in=85] (-6.15, -6);
\draw[thick] (-6.15, -6) to [out=-95, in=90] (-6.15, -7) to [out=-90, in=120](-5.85, -8.75);
\draw[ultra thick] (-6.35,-7.15) to [out=-15, in=-165] (-6.15,-7.15);
\node[right] at (-7.5,-9.3) {$x\to p_i$};

\draw[thick] (-5.15,-9.25) to (-2.35, -9.25) to (-1.35,-7.25) to (-3.1, -7.25); 
\draw[thick] (-3.35, -7.25) to (-4.15, -7.25) to (-5.15, -9.25);
\draw[thick, dashed] (-3.6, -5) to [out=-60, in=95] (-3.35, -6);
\draw[thick] (-3.35, -6) to [out=-85, in=90] (-3.35, -6.25) to [out=-90, in=60](-3.75, -8.25);
\draw[thick, dashed] (-2.75, -5) to [out=-120, in=85] (-3.15, -6);
\draw[thick] (-3.15, -6) to [out=-95, in=90] (-3.15, -6.5) to [out=-90, in=120](-2.85, -8.25);
\draw[ultra thick] (-3.35,-6.65) to [out=-15, in=-165] (-3.15,-6.65);

\draw[thick] (-2.15,-10.25) to (0.65, -10.25) to (1.65,-8.25) to (-0.1, -8.25); 
\draw[thick] (-0.35, -8.25) to (-1.15, -8.25) to (-2.15, -10.25);
\draw[thick, dashed] (-0.6, -5) to [out=-60, in=90] (-0.35, -6);
\draw[thick] (-0.35, -6) to [out=-90, in=90] (-0.35, -6.25) to [out=-90, in=60](-0.75, -9.25);
\draw[thick, dashed] (0.25, -5) to [out=-120, in=90] (-0.15, -6);
\draw[thick] (-0.15, -6) to [out=-90, in=90] (-0.15, -6.5) to [out=-90, in=120](0.15, -9.25);
\draw[ultra thick] (-0.35,-7.35) to [out=-15, in=-165] (-0.15,-7.35);

\draw[thick] (-7.05,-7.05) to (-6.35, -7.05);
\draw[thick] (-6.15, -7.05) to (-4.25, -7.05) to (-3.8, -6);
\draw[dotted] (-3.8, -6) to (-3.4,-5.2) to (-6.25, -5.2) to (-6.6, -6);
\draw[thick] (-6.6, -6) to (-7.05, -7.05);
\draw[thick, dashed] (-5.8, -3.25) to [out=-60, in=95] (-5.55, -4.25) to [out=-85, in=90] (-5.55, -4.5) to [out=-90, in=60](-5.95, -6.5);
\draw[thick, dashed] (-4.95, -3.25) to [out=-120, in=85] (-5.35, -4.25) to [out=-95, in=90] (-5.35, -4.75) to [out=-90, in=120](-5.05, -6.5);
\draw[ultra thick] (-5.55,-4.9) to [out=-15, in=-165] (-5.35,-4.9);

\draw[thick] (-4.05,-7.05) to (-3.35, -7.05);
\draw[thick] (-3.15, -7.05) to (-1.25, -7.05) to (-0.8, -6);
\draw[dotted] (-0.8, -6) to (-0.4,-5.2) to (-3.25, -5.2) to (-3.6, -6);
\draw[thick] (-3.6, -6) to (-4.05, -7.05);
\draw[thick, dashed] (-2.8, -3.25) to [out=-60, in=95] (-2.55, -4.25) to [out=-85, in=90] (-2.55, -4.5) to [out=-90, in=60](-2.95, -6.5);
\draw[thick, dashed] (-1.95, -3.25) to [out=-120, in=85] (-2.35, -4.25) to [out=-95, in=90] (-2.35, -4.75) to [out=-90, in=120](-2.05, -6.5);
\draw[ultra thick] (-2.55,-4.9) to [out=-15, in=-165] (-2.35,-4.9);

\draw[thick] (-1.05,-7.05) to (-0.35, -7.05);
\draw[thick] (-0.15, -7.05) to (1.75, -7.05) to (2.2, -6);
\draw[dotted] (2.2, -6) to (2.6,-5.2) to (-0.25, -5.2) to (-0.6, -6);
\draw[thick] (-0.6, -6) to (-1.05, -7.05);
\draw[thick, dashed] (0.2, -3.25) to [out=-60, in=95] (0.45, -4.25) to [out=-85, in=90] (0.45, -4.5) to [out=-90, in=60](0.05, -6.5);
\draw[thick, dashed] (1.05, -3.25) to [out=-120, in=85] (0.65, -4.25) to [out=-95, in=90] (0.65, -4.75) to [out=-90, in=120](0.95, -6.5);
\draw[ultra thick] (0.45,-4.9) to [out=-15, in=-165] (0.65,-4.9);

\end{tikzpicture}
\caption{Brill-Lindquist metrics.}\label{fig1}
\end{figure}

Relative to the asymptotic end where $|x|\to \infty$ one computes
$$m_{\ADM}(\infty)=\sum a_i.$$
However, relative to the asymptotic end where $x\to p_i$ one computes  
\begin{equation}\label{baremass-BL}
m_{\ADM}(p_i)=a_i\left(1+\frac{G}{2c^2}\sum_{j\neq i} \frac{a_j}{|p_i-p_j|}\right).
\end{equation}
The reader is encouraged to examine parallels between \eqref{CheatingEqn} and the expression \eqref{baremass-BL}. Once again we are lead to distinguishing the \emph{effective mass} $a_i$ from the \emph{bare mass} given by \eqref{baremass-BL}.

\subsection{On relativistic point-sources}\label{intro-point-sources}
As seen in Sections \ref{heuristics-subsection} and \ref{intro-BL}, the non-linearity of RPP is tied to the concept of interaction energy between different parts of a matter distribution. In contrast to Newtonian (linear) theory of gravity, the discrepancy between the concepts of the effective mass and the bare mass makes it at least difficult if not entirely impossible to view relativistic mass as an integral of mass density\footnote{Defining a suitable notion of quasi-local mass is still one of the most investigated problems in general relativity.}. This fact alone is an indicator that the Dirac delta framework -- the most commonly used framework for managing point-sources -- is inadequate in the relativistic context. The observation we just made has been in the literature at least since the landmark 1960-62 papers of Arnowitt, Deser and Misner \cite{ZeroMass, ADM}. Specifically, it can be shown that an employment of Dirac delta function in the absence of, say, charge necessitates the vanishing of the ADM mass of the point-source. We quote from \cite{ADM}:
\begin{quote}
.... mass only arises if a particle has nongravitational interaction ...
\end{quote}

The subject of point-sources was revisited more recently in \cite{NI} where -- in contrast to \cite{ZeroMass, ADM} -- it is shown that a gravitational point-source of non-zero mass which is not coupled to any other field (and is in particular electrically neutral) can indeed be constructed in a mathematically rigorous way. The results of \cite{NI} make it clear that we need to rethink the deeply engrained presumption that point-sources are to be modeled using the Dirac delta functions. To see what the appropriately non-linear substitute might be, let us investigate the RPP 
\begin{equation*}
\theta_{\Schw} \Delta_{\geucl} \theta_{\Schw} \dvolg= -4\pi \tfrac{G}{2c^2}\omega_{\Schw} \  \text{for} \ \theta_{\Schw} = (1 + \tfrac{G}{2c^2}\tfrac{m}{r})
\end{equation*}
for a Schwarzchild body $(1+\tfrac{G}{2c^2}\frac{m}{r})^4 \geucl$. We see that instead of (the multiple of) the Dirac delta function the bare mass $\omega_{\Schw}$ of a point-source schematically equals
\begin{equation}\label{meet-delta-over-r:eqn}
m\,\delta + m^2\,\frac{G}{2c^2}\,\frac{\delta}{r}= \omega_{\Schw}.
\end{equation}
Authors of \cite{NI} go on to make the idea of $\frac{\delta}{r}$ mathematically rigorous, and propose a point-source model based on such a $\frac{\delta}{r}$-framework. 

The approach to $\frac{\delta}{r}$ presented in \cite{NI} is reliant on blow up analysis. In a nutshell\footnote{The framework in \cite{NI} is more general and includes approximately self-similar ``collapse" to a point-source.}, the idea is to use a fixed matter distribution $\Omega$, a dilation $\H_{n}: y \mapsto x=y/n$ and a sequence of matter distributions $n\, (\H_{n})_{*}\Omega$. Note that because the Dirac delta distribution corresponds to $(\H_{n})_{*}\Omega$ it is the inclusion of the multiplicative factor of $n$ which makes $n\, (\H_{n})_{*}\Omega$ scale like $\frac{\delta}{r}$. Ultimately, the idea of \cite{NI} is to investigate the limit of geometries $\theta_n^4\geucl$, which are related to $n\, (\H_{n})_{*}\Omega$ by means of the RPP. 

A careful reader may have noticed that the description in the previous paragraph is at least somewhat flawed due to units. Metaphorically speaking, if the units on $\delta$ were to be $\mathrm{kg}$ then the units on $\frac{\delta}{r}$ would have to be $\frac{\mathrm{kg}}{\mathrm{m}}$. To address any and all concerns of this sort we now provide a precise formulation of the framework of \cite{NI}. (For a generalization adapted to different levels of gravitational interactivity see Section \ref{furthercontext:sec}.)

\begin{definition}\label{framework-basic}
Let $\Omega$ be a matter distribution supported on a compact subset of $\R^3$ and let $r_n$ be a sequence of positive numbers with $r_n\to 0$.
Consider
$\rn=\left(\tfrac{G}{2c^2}\cdot\tfrac{m}{r_n}\right)^{-1}$ where $m= \int\Omega$. A sequence of distributions $\omega_n$ is self-similar of $\Omega$-type if for the dilation $\mathcal{H}_{\rn}: y\mapsto \rn\,y$ we have 
$$\Omega=\rn\cdot \mathcal{H}_{\rn}^*\omega_n, \text{\ \ i.e.\ \ }
\omega_n = (\rn)^{-1}\cdot \left(\mathcal{H}_{\rn}\right)_*\Omega.$$
\end{definition}

As mentioned above, the idea of \cite{NI} is to investigate the limit of geometries 
$$g_n=\theta_{n}^4\geucl,$$ 
which are related to $\omega_{n}$ by means of the RPP. Ultimately, the main result of \cite{NI}, depicted in Figure \ref{fig2}, is that one indeed recovers Schwarzschild metric in the limit. 
\begin{figure}[h]
\centering
\begin{tikzpicture}[scale=.33]

\draw (-0.5,-6) to (6.25, -6) to (8.25,-4) to (1.5, -4) to (-0.5, -6);

\draw[thick, dashed] (2, -4.5) to [out=-30, in=100] (2.75, -6);
\draw[thick] (2.75, -6) to [out=-80, in=180] (3.75, -6.5) to [out=0, in=-100] (4.75, -6);
\draw[thick, dashed] (4.75, -6) to [out=80, in=-150] (5.75, -4.5);

\draw (8.75,-6) to (14, -6) to (16,-4) to (10.75, -4) to (8.75, -6);

\draw[thick, dashed] (11.5, -4.5) to [out=-30, in=100] (12.15, -5.25) to [out=-100, in=30] (11.75, -6);
\draw[thick] (11.75, -6) to [out=-150, in=90] (10.75, -7) to [out=-90, in=180] (12.25, -7.75) to [out=0, in=-100] (14.25, -7) to [out=80, in=-45] (13.5, -6);
\draw[thick, dashed] (13.5, -6) to [out=135, in=-90] (13.25, -5.25) to [out=80, in=-150] (13.75, -4.5);

\draw[thin] (12.15, -5.15) to [out=-10, in=-160] (13.25, -5.15);

\draw (16.65,-6) to (22, -6) to (24,-4) to (18.65, -4) to (16.65, -6);

\draw[thick, dashed] (19, -4.5) to [out=-30, in=100] (19.75, -5.6) to [out=-100, in=30] (19.5, -6);
\draw[thick] (19.5, -6) to [out=-120, in=75] (17.25, -7.25) to [out=-90, in=180] (19.5, -8.65) to [out=0, in=-100] (23, -7.5) to [out=100, in=-75] (20.8, -6);
\draw[thick, dashed] (20.8, -6) to [out=105, in=-90] (20.75, -5.6) to [out=80, in=-150] (21.5, -4.37);

\draw[thin] (19.75, -5.55) to [out=0, in=-165] (20.8, -5.55);

\draw (23.75, -6.5) to (26, -6.5);
\draw (25.8, -6.7) to (26, -6.5) to (25.8, -6.3);

\draw (26.75,-6) to (31.5, -6) to (33.5,-4) to (28.75, -4) to (26.75, -6);

\draw[thick, dashed] (29, -4.5) to [out=-30, in=90] (29.75, -6);
\draw[thick] (29.75, -6) to [out=-90, in=25] (28.75, -8);
\draw[thick] (31.75, -8) to [out=160, in=-90] (30.6, -6);
\draw[thick, dashed] (30.6, -6) to [out=90, in=-150] (31.45, -4.37);

\draw[thick] (29.8, -6.15) to [out=-20, in=-160] (30.6, -6.15);

\draw (26.5,-9) to (31.5, -9) to (33.5,-7) to (28.5, -7) to (26.5, -9);

\end{tikzpicture}
\caption{Pictorial description of results from \cite{NI}.}\label{fig2}
\end{figure}
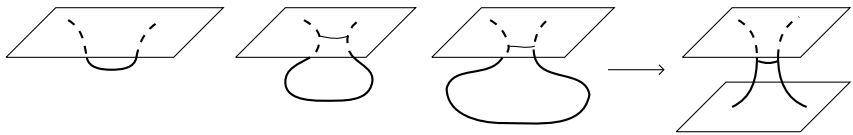

\begin{theorem}\label{NI-oldthm}
Adopt the notation established above, and fix some $k\in \mathbb{N}$. There exists an exhaustion
$$J_1\subseteq J_2\subseteq J_3\subseteq ... \subseteq \mathbb{R}^3\smallsetminus\{0\},\ \ \bigcup_{n=1}^\infty J_n=\mathbb{R}^3\smallsetminus\{0\}$$
with precompact open sets and embeddings $j_n:J_n\to \mathbb{R}^3$ such that
$$\|j_n^*g_{n}-g_{\mathrm{Schw}}\|_{C^k(J_n, g_{\mathrm{Schw}})}\to 0  \text{\ \ as\ \ } n\to \infty,$$
where $g_{\mathrm{Schw}}=\left(1+\tfrac{G}{2c^2}\tfrac{\mE}{|x|}\right)^4\geucl$ for some positive parameter $\mE$.
\end{theorem}

In addition, the mass parameter $\mE$ is explicitly computed in \cite{NI}. To state the result regarding $\mE$ we need to revisit Theorem \ref{ThetaF:prop} for the specific value of $b=0$. 
Let $\Theta$ denote the unique, smooth, positive solution of 
\begin{equation}\label{thecrux}
\Theta \Delta_{\geucl} \Theta\, \dvol_{\geucl}= -4\pi \frac{G}{2c^2} \Omega,  \ \ \ \ \lim_{|y| \to \infty} \Theta(y) = 0.
\end{equation}
The work in \cite{NI} (see Section 1.5) shows that 
$$\mE=\int \frac{\Omega}{\Theta}.$$
In hindsight, it would have been prudent to employ the following normalization in \cite{NI}. 

\begin{remark}\label{normalization:remark}
Replacing $\Omega$ with $\Omega_C = C\,\Omega$ amounts to replacing $\Theta$ with $\Theta_C = \sqrt{C}\,\Theta$. If we choose $C$ so that $$C = \left(\int\frac{\Omega}{\Theta}\right)^{-2}$$
we have the following consequence:
$$ \int \frac{\Omega_C}{\Theta_C} = \sqrt{C}\int\frac{\Omega}{\Theta} = 1.$$ 
From now on we always assume such a choice of $C$. In other words, from now on we always assume a fixed matter distribution $\Omega_0$ such that for its corresponding $\Theta_0$ we have 
$$\int \frac{\Omega_0}{\Theta_0} =1.$$
\end{remark}

While the above normalization seems algebraically innocent, it does have a somewhat profound physical interpretation.  The narrative thus far has de-facto assumed that the units on matter distributions 
$\Omega$ (see Definition \ref{matterdist:defn}) are $\mathrm{kg}$ while the conformal factors $\Theta$ are unit-less. Under such a choice, the constant $C$ from Remark \ref{normalization:remark} has the units of $\mathrm{kg}^{-2}$. Overall, such a choice of units forces the units on $\Omega_0$ and $\Theta_0$ to be $\frac{1}{\mathrm{kg}}$. Should the reader find themselves uncomfortable with this perspective, they may benefit from the following informal description. Both $\Omega_0$ and $\Theta_0$ are in some sense ``bindable" to mass: $\Omega_0$,  much like $\frac{G}{2c^2}\cdot \frac{\delta}{r}$ in \eqref{meet-delta-over-r:eqn}, awaits $(\mathrm{mass})^2$ to be multiplied with. At the same time, $\Theta_0$ awaits just $\mathrm{mass}$. For clarity reasons we emphasize  that $\Theta=\Theta_0$ solves 
$$\Theta \Delta_{\geucl} \Theta\, \dvolg= -4\pi \frac{G}{2c^2} \Omega_0,  \ \ \ \ \lim_{|y| \to \infty} \Theta(y) = 0,$$
while $\Theta=m\Theta_0$ solves
\begin{equation}\label{m^2:eqn}
\Theta \Delta_{\geucl} \Theta\, \dvolg= -4\pi \frac{G}{2c^2} m^2\Omega_0,  \ \ \ \ \lim_{|y| \to \infty} \Theta(y) = 0.
\end{equation}
In hindsight, it would have been prudent to center the work in \cite{NI} around \eqref{m^2:eqn} rather than \eqref{thecrux}. In this paper we assume normalized $\Omega_0$. 

\subsection{On superposition of relativistic point-sources}

In the sense discussed in Sections \ref{heuristics-subsection}, \ref{intro-BL} and \ref{intro-point-sources} above, the non-linearity of general relativity makes it questionable if one should even attempt to view a relativistic cloud of matter as being made up of point-source constituents. And yet, there is an undeniable sense that mass, much like volume or simply amount of some substance, is an \underline{\emph{extensive}} variable. In this paper we point to the possibility that relativistic clouds of matter could in fact be viewed as cumulative effects of point-sources provided one interprets the concept of point-sources as in \cite{NI}.

Perhaps the most delicate issue to circumnavigate is the inherent lack of superposition associated with the RPP. Even if one were to believe that the framework of \cite{NI} leads to the appropriate interpretation of the phrase ``point-source", there is no reason to believe that the framework lends itself to a principle of superposition. It is natural to think that Brill-Lindquist metrics \eqref{BL-maineqn} might somehow embody the idea of superposition, but it is not at all obvious (and it is maybe not even expected) that such an intuitive understanding can be put on mathematically rigorous footing. The entire goal of our paper is to provide a framework which serves as a de-facto superposition of point-sources of \cite{NI} and which recovers Brill-Lindquist metrics \eqref{BL-maineqn} in the limit. Specifically, our framework relies on the following definition. 

\begin{definition}[Source configuration]\label{dust:defn}
Suppose $\Omega_0$ is a normalized matter distribution in the sense of Remark \ref{normalization:remark}, and suppose $\Omega_0$ is supported on the unit ball centered at the origin. Let $\P = \{ (p_1, a_1), ... , (p_{Q}, a_{Q})   \}$ be a set of $Q$ point-sources, each located at $p_i$ with effective mass $a_i$. Finally, consider the dilations 
$$\H_{n, p_i}: y \mapsto x = p_i + y/n.$$ 
We say $\omega_{\P,n}$ is a source configuration of $(\Omega_0, \P)$-type if
\begin{equation*}
\omega_{\P,n} =\sum_{i=1}^{Q} a_i^2\cdot n \, (\H_{n,p_i})_{*}\Omega_0.
\end{equation*}
By the separation parameter of the configuration we mean
$$\sigma (\P):=\min \{  |p_j - p_i | \big | i \neq j \}.$$
When $\P$ is clear from context we simply write $\omega_n$ and $\sigma$. 
\end{definition}

Much as in \cite{NI} we analyze solutions $\theta_{\P,n}$ to 
\begin{equation}\label{littleRPP:eqn}
\theta_{\P,n} \Delta_{\geucl} \theta_{\P,n}\dvolg = -4\pi \tfrac{G}{2c^2} \omega_{\P,n}, \  \lim_{|x| \to \infty} \theta_{\P,n}(x) = 1,
\end{equation}
where $\omega_{\P,n}$ is from Definition \ref{dust:defn}. The following is our main result. 

\begin{theorem}[Superposition Theorem]\label{Will-Iva-Thm1}
We have 
$$\theta_{\P,n}(x)\to 1+ \frac{G}{2c^2}\sum_{i=1}^{Q}  \frac{a_i}{|x -p_i|}\ \ \text{as}\ \ n\to \infty$$
over all compact subsets of $\R^3\smallsetminus \{p_1, ..., p_Q\}$. The convergence is uniform with all the derivatives.
\end{theorem}

The limit function in Theorem \ref{Will-Iva-Thm1} takes the form of the conformal factor of the Brill-Lindquist metrics for the cloud of point-sources of effective masses $a_i$ located at $p_i$ (see Section \ref{intro-BL}). This fact motivated the employment of the term ``effective mass" in Definition \ref{dust:defn}.

\subsection{One-parameter family of non-linear superposition principles}\label{furthercontext:sec}
We would also like to present the case that our superposition framework, made precise in Definition \ref{dust:defn} and Theorem \ref{Will-Iva-Thm1} above, is in fact a continuous extension of the classical superposition principle of Newtonian gravity. Inspired by the observation that the coupling of $\theta$ and $\Delta_{\geucl}\theta$ in \eqref{TheEqn} is modeling gravitational interaction (see Section \ref{heuristics-subsection}) we introduce the following one-parameter family of \emph{Generalized Poisson Problems}.

\begin{definition}[GPP]\label{GPP:defn}
Let $\alpha\in[0,1]$. By a Generalized Poisson Problem (GPP) we mean 
$$\theta^\alpha \Delta_{\geucl}\theta \dvol_{\geucl}=-4\pi \tfrac{G}{2c^2} \omega,\ \ \lim_{x\to \infty} \theta(x)=b$$
with $b\ge 0$. When $b$ is not explicitly mentioned the reader should assume the value of $b=1$.
\end{definition}

The idea here is that the coupling of $\theta^\alpha$ to $\Delta_{\geucl}\theta$ models gravitational interactivity. The continuous parameter $\alpha$ marks the departure from the Newtonian Poisson equation ($\alpha=0$) towards its relativistic counterpart ($\alpha=1$). Basically, the introduction of the parameter $\alpha$ allows us to \emph{continuously} transition from the non-interactivity of the Newtonian matter towards the full interactivity of the relativistic matter. The following is the main theorem addressing the existence and the uniqueness of the solutions of the GPP, and is proven in Section \ref{GPP:sec}. The theorem also addresses the continuity of the GPP framework with respect to the parameter $\alpha$.  

\begin{theorem}\label{AngryIvaThm}
For each $b\in[0,1]$ and each $\alpha\in[0,1]$ there exists a unique solution $\theta_{b,\alpha}$ of 
\begin{equation}\label{TheAlphaEqn-modified}
\theta_{b,\alpha}^\alpha \Delta_{\geucl}\theta_{b,\alpha} \dvol_{\geucl}=-4\pi \tfrac{G}{2c^2} \omega,\ \ \lim_{x\to \infty} \theta_{b,\alpha}(x)=b,
\end{equation}
which in addition satisfies the asymptotic conditions \eqref{asymptotic-conditions}. The family $\theta_{b,\alpha}$ depends continuously on $(b,\alpha)$ in the sense that for all convergent sequences $(b_n,\alpha_n)\to (b,\alpha)$ in the permissible range we have convergences $\theta_{b_n,\alpha_n}\to \theta_{b,\alpha}$ with all derivatives on all compact subsets of $\R^3$.
\end{theorem}

Of particular interest for our paper is the solution $\theta_{0,\alpha}$ of the GPP corresponding to the value of $b=0$. The following definition goes in parallel with Remark  \ref{normalization:remark}. 
\begin{definition}[$\alpha-$normalization]\label{newnormalization:definition}
By an $\alpha$-normalized matter distribution we mean a matter distribution $\Omega_{0,\alpha}$ supported on the standard unit ball such that 
$$\int \frac{\Omega_{0,\alpha}}{\Theta_{0,\alpha}^\alpha}=1.$$
Here $\Theta_{0,\alpha}$ denotes the solution of the GPP
$$\Theta^\alpha \Delta_{\geucl} \Theta \dvolg =-4\pi \frac{G}{2c^2} \Omega_{0,\alpha},\ \ \lim_{|y|\to \infty} \Theta(y)=0$$
whose existence and uniqueness is established in Theorem \ref{AngryIvaThm}. When a value of $\alpha$ is clear from context\footnote{E.g. the value of $\alpha=1$ is used throughout the sections of our paper dealing with RPP.} we drop $\alpha$ from the notation and simply write $\Omega_0$ and $\Theta_0$.
\end{definition}
The existence of $\alpha$-normalized distributions, as well as the fact that focusing only on them does not reduce any generality, can be established much as in Remark \ref{normalization:remark}; the appropriate value of $C$ in the generalized framework is 
$$C=\left(\int \frac{\Omega}{\Theta^\alpha}\right)^{-(1+\alpha)}.$$
It is interesting to observe that, for a fixed distribution $\Omega$ and by virtue of Theorem \ref{AngryIvaThm}, the normalization constant $C$ varies continuously in $\alpha$. Thus, there is a sense in which $\alpha$-normalization is continuous in $\alpha$.
We are now in position to alter Definition \ref{dust:defn} and make it compatible with the framework of GPP. 

\begin{definition}\label{sourceconfiga}
Let $\P$ and $\H_{n, p_i}$ be as in Definition \ref{dust:defn} and let $\Omega_{0,\alpha}$ be an $\alpha$-normalized matter distribution in the sense of Definition \ref{newnormalization:definition}. We say $\omega_{\P,\alpha,n}$ is a source configuration of $(\alpha, \Omega_{0,\alpha}, \P)$-type if 
\begin{equation}\label{alpha-matter:defn}
\omega_{\P,\alpha,n} =\sum_{i=1}^{Q} a_i^{\alpha+1}\cdot n^{\alpha} \, (\H_{n,p_i})_{*}\Omega_{0,\alpha}.
\end{equation}
When $\alpha$ and $\P$ are clear from context we simply write $\omega_n$.
\end{definition}

Note that in the situation when $\alpha=0$ the configuration \eqref{alpha-matter:defn} reduces to $\sum_{i=1}^{Q} a_i\cdot (\H_{n,p_i})_{*}\Omega_{0}$. The latter is akin to a discrete version of the decomposition 
$$\varrho(x)=\int_{p\in \R^3} \varrho\big{|}_p\cdot \delta_p(x),$$
where $\delta_p$ denotes the Dirac delta distribution centered at $p$. On the other hand, the value of $\alpha=1$ corresponds to the relativistic matter distribution of Definition \ref{dust:defn}. As discussed above, \eqref{alpha-matter:defn} exhibits continuity in $\alpha$ and can thus be interpreted as providing a continuous transition from the classical Newtonian framework to the relativistic framework of the RPP. 

For the value of $\alpha=0$ the following result is recognizable as the classical Superposition Principle. By permitting $\alpha$ to vary we obtain a one-parameter family of superposition principles. 

\begin{theorem}[Generalized Superposition Theorem]\label{Will-Iva-alphaThm}
Consider the solutions $\theta_{\P,\alpha, n}$ to 
\begin{equation}\label{littleGPP:eqn}
\begin{cases}
\theta^{\alpha}_{\P,\alpha,n} \Delta_{\geucl} \theta_{\P,\alpha,n} \dvolg= -4\pi \tfrac{G}{2c^2} \omega_{\P,\alpha,n},\\
\lim_{|x| \to \infty} \theta_{\P,\alpha,n}(x) = 1.
\end{cases}
\end{equation}
where $\omega_{\P,\alpha,n}$ is from Definition \ref{sourceconfiga}. We have 
$$\theta_{\P,\alpha, n}(x)\to 1+ \frac{G}{2c^2}\sum_{i=1}^{Q}  \frac{a_i}{|x -p_i|}\ \ \text{as}\ \  n\to \infty$$
over all compact subsets of $\R^3\smallsetminus \{p_1, ..., p_Q\}$. The convergence is uniform with all the derivatives.
\end{theorem}

For the value of $\alpha=1$ Theorem \ref{Will-Iva-alphaThm} reduces to Theorem \ref{Will-Iva-Thm1}. It is in this sense that we understand Theorem \ref{Will-Iva-Thm1} as a continuous extension of the classical superposition principle of Newtonian gravity. Theorem \ref{Will-Iva-alphaThm} itself is proven in Section \ref{GPP:sec}.

\subsection*{Acknowledgments}
Our research has been funded by John S. Rogers Science Research Program at Lewis \& Clark College.

\section{Proof of the Superposition Theorem \ref{Will-Iva-Thm1}.}

\subsection{Review of linear theory}\label{LinTheory:sec}
For a constant $b\in \R$ and a compactly supported smooth function $\varrho$ on $\R^3$ the integral $\int_{\xi \in\R^3}\frac{\varrho(\xi)}{|x-\xi|}\,\dvol_{\xi}$ is absolutely convergent and defines a function 
\begin{equation}\label{GRF-verN}
u(x)=b+\int_{\xi \in\R^3}\frac{\varrho(\xi)}{|x-\xi|}\,\dvol_{\xi}.
\end{equation}
Under the smoothness assumption on $\varrho$ one sees that $u$ is itself smooth, although differentiation of \eqref{GRF-verN} under the integral sign is only appropriate in certain situations. More specifically, differentiation under the integral sign is justified in the case of the first derivatives and in the case of $x\not \in \mathrm{supp}(\varrho)$. 

The function $u(x)$ of \eqref{GRF-verN} is the unique solution to the asymptotic boundary value problem 
$$\Delta_{\geucl} u = -4\pi \varrho,\ \ \lim_{|x|\to \infty} u(x)=b$$ 
and satisfies the asymptotic conditions \eqref{asymptotic-conditions}.
 
We now record boundedness properties of $u(x)$ needed in our paper. The proof is elementary and left to the reader. 

\begin{lemma}\label{WillsLemma3}
\ 
\begin{enumerate}
\item There is a constant $C_+$ which depends only on $\mathrm{supp}(\varrho)$ such that 
$$\|u\|_{L^\infty(\R^3)}\le |b|+ C_+\|\varrho\|_{L^\infty(\R^3)}.$$
\medbreak
\item For each fixed non-negative $\varrho$ and a compact set $K\subseteq \R^3$ there is a positive constant $C_-=C_-(K)>0$ such that 
$$u(x) \ge b+ C_-(K) \text{\ \ for all\ \ }x\in K.$$
\end{enumerate}
\end{lemma}

The reader should note that the representation formula 
\begin{equation}\label{GRF-main}
u(x)=b-\frac{1}{4\pi}\int_{\xi \in \R^3}\frac{\Delta_{\geucl} u(\xi)}{|x-\xi|}\,\dvol_{\xi}
\end{equation}
holds more generally -- even in situations when $\varrho=-\tfrac{1}{4\pi}\Delta u$ is not compactly supported. Indeed, one can show that \eqref{GRF-main} holds whenever $u$ satisfies the asymptotic decay conditions \eqref{asymptotic-conditions}.

\subsection{Strategy}
The main strategy in the proof of Theorem \ref{Will-Iva-Thm1} is to use pullback under $\H _{n, p_j}$ where, as in Definition \ref{dust:defn}, 
$$\H_{n, p_j}: y \mapsto x = p_j + y/n.$$ 
We now introduce the notation for (scaled) pullbacks we use throughout our proof. 

\begin{definition}\label{pullb:defn} 
Adapting the notation from Definition \ref{dust:defn}, we define
$$\Omega_{n,p_j} =\tfrac{1}{n}\H_{n,p_j}^*\omega_{\P,n} \text{\ \ and\ \ }
\Theta_{n, p_j} = \tfrac{1}{n}\H _{n, p_j}^{*}\theta_{\P,n}.$$
In addition, define the functions $\Phi_0$ and $\Phi_{n,p_j}$ by 
$$\Omega_0=\Phi_0\,\dvol_{\geucl} \text{\ \ and\ \ }\Omega_{n,p_j}=\Phi_{n,p_j}\,\dvolg.$$
\end{definition}

Let us take the moment to record explicit expressions for $\Omega_{n,p_j}$ and $\Phi_{n, p_j}$.

\begin{remark}\label{pullb}
It follows from $\H _{n, p_i}^{-1}\circ \H_{n,p_j}(y)=y-n(p_i-p_j)$ and Definition \ref{dust:defn} that 
\begin{equation*}\begin{aligned}
\Omega_{n, p_j}(y) = &\H _{n, p_j}^*\left(\sum_{i=1}^{Q} a_i^2\cdot (\H_{n,p_i})_{*} \Omega_0(y) \right)\\
=& \sum_{i=1}^{Q}  a_i^2\cdot \Omega_0(y- n(p_i -p_j)),
\end{aligned}\end{equation*}
Consequently, we have $\Phi_{n,p_j}(y)=  \sum_{i=1}^{Q} a_i^2 \Phi_0(y - n(p_i - p_j))$.
\end{remark}

Since $\H _{n, p_j}^*\Delta_{\geucl}=n^2\Delta_{\geucl}$ and $\H _{n, p_j}^*\dvol_{\geucl}=\tfrac{1}{n^3}\dvol_{\geucl}$, pulling \eqref{littleRPP:eqn} back under $\H _{n, p_j}$ yields 
\begin{equation}\label{BigRPP:eqn}
\Theta_{n,p_j} \Delta_{\geucl} \Theta_{n,p_j} = -4\pi \frac{G}{2c^2}\Phi_{n, p_j}, \  \lim_{|y| \to \infty} \Theta_{n, p_j}(y) = \tfrac{1}{n}
\end{equation} 
with $\Phi_{n,p_j}$ and $\Theta_{n, p_j}$ as in Definition \ref{pullb:defn}. 

The key to proving Theorem \ref{Will-Iva-Thm1} is in showing that for each fixed $p_j$ the sequence of functions $\Theta_{n, p_j}$  converges to $a_j\Theta_0$ as $n\to \infty$. 
It is because of this that we are mainly interested in $\omega_{\P,n}$ for substantially large values on $n$. For example, we always assume that $n$ is large enough so that $\frac{1}{n}\ll \sigma(\P)$ and, specifically, 
$$i\neq j \longrightarrow \mathrm{supp} \left((\H_{n,p_i})_{*} \Omega_0\right)\cap \mathrm{supp} \left((\H_{n,p_j})_{*} \Omega_0\right)=\emptyset.$$
In addition, we frequently make use of the following observation: For a given compact set $K\subseteq \R^3$ there exists $N(K,\sigma)$, depending only on $K$ and $\sigma$, such that for all $n\ge N(K,\sigma)$ we have
$$y\in K,\ i\neq j \longrightarrow y-n(p_i-p_j)\not\in B(0,1).$$

In view of the assumption that $\mathrm{supp}(\Omega_0) = B(0,1)$ (see Definition \ref{dust:defn}) expressions of Remark \ref{pullb} now give us the following.

\begin{lemma}\label{ivaslastminutebs}
Let $K$ be a compact subset of $\R^3$. There exists $N(K,\sigma)\in \mathbb{N}$ such that $\Omega_{n, p_j}=a_j^2 \Omega_0$ for all $n\ge N(K,\sigma)$ and all $p_j$. In particular, we have 
$$\Omega_{n, p_j} \to a_j^2 \Omega_0 \text{\ \ as\ \ } n\to \infty.$$
This convergence is uniform with all derivatives over all compact subsets of $\R^3$.
\end{lemma}

For convenience of the reader we also record another consequence of Remark \ref{pullb}. Note that the second of the claims relies on smoothness of $\Phi_0$ and the fact that expressions in Remark \ref{pullb} involve evaluation at $y-n(p_i-p_j)$.

\begin{lemma}\label{hboundp}
We have
\begin{enumerate}
\item $\|\Phi_{n,p_j}\|_{L^\infty(\R^3)}\leq \left(\sum a_i^2\right)\|\Phi_0\|_{L^\infty(\R^3)}$ and
\medbreak 
\item $\|\Phi_{n,p_j}\|_{H^k(\R^3)}\leq \left(\sum a_i^2\right)\|\Phi_0\|_{H^k(\R^3)}$ for all $k\ge 0$.
\end{enumerate}
In particular, the bounds mentioned above are independent of $n$ and our choice of $p_j$. 
\end{lemma}

\subsection{Bounds on $\Theta_{n, p_j}$}\label{Bounds2}
As stated earlier, our strategy is to show the sequence $\Theta_{n, p_j}$ converges to $a_j \Theta_0$. We will use a diagonal argument to construct a convergent subsequence of $\Theta_{n, p_j}$. The iterative process within the diagonal argument relies on Rellich Lemma. Thus, we continue by establishing bounds on $\Theta_{n, p_j}$.

\begin{lemma}\label{ThetaBounds:prop}\ 
\begin{enumerate}
\item The functions $\Theta_{n, p_j}$ are bounded in $L^\infty(\R^3)$, independently of $n$ and our choice of $p_j$:
$$\|\Theta_{n, p_j}\|_{L^\infty(\R^3)}\le M_+ \text{\ \ for all\ \ } n, p_j.$$
\medbreak
\item Given a compact set $K$ there exists a positive constant $M_-=M_-(K)$ so that 
\begin{equation*}
\Theta_{n, p_j}(y) \ge a_j^2\cdot M_{-}(K) \ \ \text{for all $n, p_j $ and all $y\in K$}.
\end{equation*}
\end{enumerate}
\end{lemma}

\begin{proof}
We begin by bounding $\Theta_{n, p_j}$ from the above. Consider the fact that 
\begin{equation}\label{keyinequality}
\Delta_{\geucl}(\Theta_{n, p_j}^2) = 2 \Theta_{n, p_j} \Delta_{\geucl} \Theta_{n, p_j} + 2 |d\Theta _{n, p_j}|^2 \geq -4\pi \frac{G}{c^2}\Phi_{n, p_j}.
\end{equation}
Since $\Theta_{n, p_j}$ satisfies \eqref{asymptotic-conditions} so does $\Theta_{n, p_j}^2$. Thus \eqref{GRF-main} still applies to $\Theta_{n,p_j}^2$. Combined with the inequality above we have,
\begin{equation}\label{BigDeal}
\begin{aligned}
\Theta_{n, p_j}^2(y) =  &\left(\frac{1}{n}\right)^2 - \frac{1}{4\pi} \int_{\xi} \frac{1}{|y - \xi|}\Delta_{\geucl }(\Theta^2_{n,p_j})(\xi) \dvol_{\geucl} \\
\leq  &\,\left(\frac{1}{n}\right)^2  + \frac{G}{c^2}\int_{\xi} \frac{\Omega_{n,p_j }(\xi)}{|y-\xi|}.
\end{aligned}
\end{equation}
It now follows from Lemmas \ref{WillsLemma3} and \ref{hboundp} that 
$$\|\Theta_{n, p_j}^2\|_{L^\infty(\R^3)} \leq \,1 + C_+\cdot \left(\sum a_i^2\right)\,\|\Phi_0\|_{L^\infty(\R^3)},$$
for some universal constant $C_+$. In particular, it follows that the functions $\Theta_{n,p_j}$ are bounded in $L^\infty(\R^3)$, independently of $n$ and our choice of $p_j$. 

We move on to show $\Theta_{n, p_j}$ are bounded away from zero over compact subsets of $\R^3$. Let $M_+$ be the $L^\infty$-upper bound we established thus far. From Green's representation formula \eqref{GRF-main} we have 
\begin{equation*}
\Theta_{n, p_j}(y) \geq  \frac{G}{2c^2}\int_{\xi} \frac{\Omega_{n, p_j}(\xi)}{|y-\xi|\Theta_{n, p_j}(\xi)} \geq  \frac{G}{2c^2}\,\frac{a_j^2}{M_+}\int_{\xi} \frac{\Phi (\xi)}{|y-\xi|}\dvolg.
\end{equation*}
Our claim is now an immediate consequence of the lower bound discussed in Lemma \ref{WillsLemma3}.
\end{proof}

Next, we address the behavior of $\Theta_{n, p_j}(y)$ as $y$ approaches infinity.

\begin{lemma}\label{ThetaDBC:prop}
Let $\e>0$, and let $p_j$ be fixed. There exists $L>1$ such that for all $|y| > L$ there exists $N \in \mathbb{N}$ so that when $n\ge N$ we have
$$\Theta_{n, p_j}(y) < \epsilon.$$
\end{lemma}
\begin{proof}
Let $L>2$ be large enough so that 
$$\frac{G}{c^2}\cdot\frac{2}{L}\,a_j^2\, \|\Omega_0\|_{L^1(\R^3)} <\frac{\e^2}{2}.$$
Choose $N$ large enough so that  
$$\left(\frac{1}{N}\right)^2 +\frac{G}{c^2}\cdot\frac{4\sigma}{N}\left(\sum a_i^2\right)\,\|\Omega_0\|_{L^1(\R^3)}<\frac{\e^2}{2}.$$
Let $y$ be such that $|y|>L$. Further increase $N$ so that 
$$L<|y|<N\sigma/2.$$ 
Let $\eta\in B(0, 1)$. Note that 
\begin{equation}\label{1/l}
\frac{1}{|y - \eta|} \leq \frac{1}{|y| - |\eta|} \leq \frac{1}{L - 1} \leq \frac{2}{L}.
\end{equation}
for all $n\ge N$. In addition, if $i\neq j$ we also have that 
\begin{equation}\label{1/n}
\begin{aligned}
\frac{1}{|y + n(p_i -p_j) - \eta|} \leq &\,\frac{1}{|n(p_i-p_j)|- |y| - |\eta|}\\  
\leq &\, \frac{1}{(n\sigma/2) - 1}\le \frac{4/\sigma}{n}.
\end{aligned}
\end{equation}

Arguing as in the proof of Lemma \ref{ThetaBounds:prop} (e.g see \eqref{BigDeal}) and employing the change of coordinates $\xi =\eta-n(p_i-p_j)$ yields
\begin{align*}
\Theta_{n, p_j}^2(y) \leq &\left(\frac{1}{n}\right)^2  +  \frac{G}{c^2}\int_{\xi} \frac{\Omega_{n, p_j }(\xi)}{|y-\xi|}  \\ 
= &\left(\frac{1}{n}\right)^2 + \frac{G}{c^2}\sum_{i=1}^{Q} a_i^2 \int_{\eta}  \frac{\Omega_0 (\eta)}{|y + n(p_i - p_j)-\eta|}.
\end{align*}
We continue by distinguishing the cases of $i=j$ and $i\neq j$ within the summation. Employing \eqref{1/l} and \eqref{1/n} yields 
$$\Theta_{n, p_j}^2(y) \leq \left(\frac{1}{n}\right)^2 +  \frac{G}{c^2}\left( a_j^2 \cdot \frac{2}{L} + \sum_{i \neq j} a_i^2 \cdot \frac{4/\sigma}{n}\right)\cdot \|\Omega_0\|_{L^1(\R^3)}.$$
Our choice of $L$ and $N$ is made so that 
$$\Theta_{n, p_j}^2(y) \leq \frac{\e^2}{2}+\frac{\e^2}{2}=\e^2.$$
This completes our proof.
\end{proof}

\subsection{Convergence of $\Theta_{n,p_j}$}

\begin{lemma}\label{ExUnT:prop}
For a fixed $p_j$ the sequence $\Theta_{n, p_j}$ converges uniformly with all derivatives over all compact subsets of $\R^3$ to $a_j\Theta_0$.
\end{lemma}

\begin{proof} The proof of the lemma is a multi-step process. First we construct a convergent subsequence of $\Theta_{n,p_j}$ using a diagonal argument. Next, we record important characteristics about the diagonal sequence's limiting function. Finally, we argue by contradiction that the full sequence converges to said limit function.  

Fix a chain of compact subsets of $\R^3$ such that 
\begin{equation}\label{Ksub}
K_0 \subseteq \text{Int}(K'_0) \subseteq K'_0 \subseteq K_1 \subseteq \text{Int}(K'_1) \subseteq K'_1 \subseteq ... \subseteq \R^3,
\end{equation}
and $\bigcup _i K_i = \R^3$. From the Interior Elliptic Regularity \cite{Jost} we have
\begin{equation}\label{H2Theta}
\norm{\Theta_{n,p_j}}_{H^2(K'_0)} \cle \norm{\Phi_{n, p_j } \Theta_{n, p_j}^{-1}}_{L^2(K_1)} + \norm{\Theta_{n,p_j}}_{L^2(K_1)}.
\end{equation}
Boundedness of $\Theta_{n,p_j}$ and $\Theta_{n,p_j}^{-1}$ in $L^2(K_1)$ follow from Lemma \ref{ThetaBounds:prop} while boundedness of $\Phi_{n, p_j}$ in $L^2(K_1)$ follows from Lemma \ref{hboundp}. Overall, we have boundedness of $\Theta_{n,p_j}$ in $H^2(K'_0)$. Observe that we in addition have 
\begin{equation*}
\norm{\Theta_{n,p_j}}_{H^4(K_0)} \cle \norm{\Phi_{n, p_j } \Theta_{n, p_j}^{-1}}_{H^2(K'_0)} + \norm{\Theta_{n,p_j}}_{L^2(K'_0)}.
\end{equation*}
Here, the boundedness of $\Phi_{n, p_j }$ and $\Theta_{n, p_j}^{-1}$ in $H^2(K'_0)$ follows from Lemma \ref{hboundp}, Lemma \ref{ThetaBounds:prop} and \eqref{H2Theta}. We are now able to conclude boundedness of $\Theta_{n,p_j}$ in $H^4(K_0)$. In fact, bootstrapping like this gives boundedness of $\Theta_{n, p_j}$ in $H^k(K_0)$ for arbitrarily large $k$. 

With the intention of using a diagonal argument, we now inductively construct convergent subsequences $\Theta_{n, p_j}^{(i)}$ from $\Theta_{n, p_j}$ with repeated applications of Rellich Lemma across the chain of compact subsets \eqref{Ksub}. We construct the subsequences $\Theta_{n, p_j}^{(i)}$ so that they converge in $H^4(K_i) \subseteq C^2(K_i)$ and are subsequences of the previous $\Theta_{n, p_j}^{(i-1)}$.

By construction, the diagonal subsequence $\Theta_{n, p_j}^{(n)}$ converges uniformly with two derivatives over all compact subsets of $\R ^3$. Denote the limit function by $\Theta_\infty$. Next, we argue that 
\begin{equation}\label{limitbc}
\lim_{|y| \to \infty} \Theta_{\infty}(y) =0.
\end{equation}
To that end let $\e>0$. By Lemma \ref{ThetaDBC:prop} there is $L>0$ and an $N\in \mathbb{N}$ such that for all $|y|>L$ and all $n\ge N$ we have 
$$\Theta_{n,p_j}^{(n)}(y)<\e.$$ 
Upon taking the limit as $n\to \infty$ we obtain 
$$\Theta_\infty(y)\le \e,$$
proving \eqref{limitbc}.

Applying Lemma \ref{ivaslastminutebs} to RPP satisfied by $\Theta_{n,p_j}^{(n)}$ we see that
$\Theta_\infty$ satisfies
\begin{equation*}
\Theta_{\infty} \Delta_{\geucl} \Theta_{\infty} \dvol_{\geucl} = -4\pi \tfrac{G}{2c^2} a_j^2\Omega_0,  \ \lim_{|y| \to \infty} \Theta_{\infty}(y) =0.
\end{equation*}
Furthermore, Theorem \ref{ThetaF:prop} states that solutions to the above are unique. In fact, we have already determined (see  \eqref{m^2:eqn} above) that $a_j\Theta_0$ is this unique solution and therefore $\Theta_\infty=a_j \Theta_0$.

It remains to show that the full sequence $\Theta_{n, p_j}$ converges to $a_j \Theta_0$ over all compact subsets of $\R^3$. To see this suppose towards a contradiction that there exists some compact set $K\subseteq \R^3$, some $\epsilon_0 > 0$ and a subsequence $\Theta_{n_{k},p_j}$ such that 
\begin{equation}\label{rnd:est}
\norm{\Theta_{n_{k},p_j} - a_j \Theta_0}_{L^{\infty}(K)} \ge \epsilon_0
\end{equation}
for every $k$. We have already shown that a subsequence of the sequence of solutions $\Theta_{n_{k}, p_j}$ of the RPP corresponding to $\Omega_{n_{k}, p_j}$ can be constructed so that $\Theta_{n_{k}, p_j}$ converges to $a_j \Theta_0$. This contradicts \eqref{rnd:est} and proves that $\Theta_{n, p_j} \to a_j \Theta_0$ in $L^{\infty}(K)$ for all compact K. To show the convergence is uniform with all the derivatives, we perform the induction on $l$ in the Interior Elliptic Regularity Estimate
\begin{align*}
&\norm{\Theta_{n, p_j} - a_j \Theta_0}_{H^{l+2}(K)} \\ & \cle \norm{\Phi_{n, p_j}\Theta_{n, p_j}^{-1}-a_j\Phi_0\Theta_0^{-1}}_{H^l(K')} + \norm{\Theta_{n, p_j} - a_j\Theta_0}_{L^2(K')}.
\end{align*}
Our proof is now complete. 
\end{proof}

We use the convergence $\Theta_{n, p_j} \to a_j \Theta_0$ to provide the proof of Theorem \ref{Will-Iva-Thm1} in Section \ref{TheProof:sec}. We end this section by recording two additional convergences. 

\begin{lemma}\label{niceone:lemma}
Fix $r_0 > 0$ and let $|x| \geq  r_0$. For each multiindex $l$ and the corresponding partial derivative $\partial^l$ we have the following $L^\infty$-convergence of functions of $\nu \in B(0, 1)$:
\begin{equation*}
\left\|\partial_x^l\left(\frac{1}{|x - \nu/n|}\right) - \partial_x^l\left(\frac{1}{|x|}\right)\right\|_{L^\infty(B(0,1))}\to 0\ \ \text{as}\ \ n\to \infty.
\end{equation*}
Furthermore, the stated convergences are uniform with respect to $x\in \R^3\smallsetminus B(0,r_0)$.
\end{lemma}

The proof of Lemma \ref{niceone:lemma} is an elementary consequence of the triangle inequality 
\begin{equation*}
\frac{1}{|x - \nu/n|} - \frac{1}{|x|} = \frac{|x|-|x- \nu/n|}{|x||x - \nu/n|} \leq \frac{|\nu|/n}{|x|\cdot(|x| - |\nu|/n)},
\end{equation*}
which for sufficiently large $n$ implies  
$$\frac{1}{|x - \nu/n|} - \frac{1}{|x|}\leq \frac{2}{nr_0^2},$$
and as such is left to the reader. It is now a corollary of Lemma \ref{ExUnT:prop} and Lemma \ref{niceone:lemma} that 
\begin{equation}\label{WillsOldLemma}
\left\|\partial_x^l\left(\frac{\Omega_0(\nu)}{|x - \nu/n|\Theta_{n, p_j}(\nu)}\right) - \partial_x^l\left(\frac{\Omega_0(\nu)}{|x|\cdot a_j\Theta_0(\nu)}\right)\right\|_{L^\infty}\to 0
\end{equation}
as $n\to \infty$. Just as in Lemma \ref{niceone:lemma}, the convergences are uniform with respect to $x\in \R^3\smallsetminus B(0,r_0)$. It is really in the form of \eqref{WillsOldLemma} that the convergence result of Lemma \ref{ExUnT:prop} is used in the proof of Theorem \ref{Will-Iva-Thm1}.

\subsection{Proof of Theorem \ref{Will-Iva-Thm1}}\label{TheProof:sec}

Recall that this theorem addresses the limit behavior of the sequence $\theta_{\P,n}$ of solutions of 
$$\theta_{\P,n} \Delta_{\geucl} \theta_{\P,n} = -4\pi \tfrac{G}{2c^2} \omega_{\P,n}, \  \lim_{|x| \to \infty} \theta_{\P,n}(x) = 1,$$
where $\omega_{\P,n}$ is as in Definition \ref{dust:defn}.

\begin{proof} 
Fix a compact subset $K$ of $\R^3\smallsetminus\{p_1, .., p_Q\}$ and let $x\in K$. 
We analyze $\theta_{\P,n}(x)$ by applying the Green's representation formula \eqref{GRF-main} and by expanding $\omega_{\P,n}$ according to Definition \ref{dust:defn}:
\begin{equation}\label{GRF-applied}
\begin{aligned}
\theta_{\P,n}(x)=&1+\frac{G}{2c^2}\int_{\xi} \frac{\omega_{\P,n}(\xi)}{|x-\xi|\theta_{\P,n}(\xi)}\\ =&1+\frac{G}{2c^2}\sum_{i=1}^{Q} a_i^2\int_{\xi} \frac{(\H_{n,p_i})_{*}\Omega_0(\xi)}{|x-\xi|\cdot \frac{1}{n}\theta_{\P,n}(\xi)}.
\end{aligned}
\end{equation}
Next, we study each term in the summation \eqref{GRF-applied} individually, with the intention of using a $\frac{\epsilon}{\sum a_i^2}$-argument at the very end.
Changing coordinates according to $\H_{n,p_i}$, by which we mean setting $\xi = p_i + \nu/n$, yields
\begin{equation*}
\int_{\xi} \frac{(\H_{n,p_i})_{*}\Omega_0(\xi)}{|x-\xi|\cdot \frac{1}{n}\theta_{\P,n}(\xi)}=
\int_{\nu} \frac{\Omega_0(\nu)}{|x-p_i - \nu/n|\Theta_{n, p_i}(\nu)}.
\end{equation*}
Now consider the fact that due to the normalization on $\Omega_0$ (see Remark \ref{normalization:remark}) we have
$$\frac{1/a_j}{|x -p_j|}=\int_{\nu} \frac{\Omega_0(\nu)}{|x -p_j|\cdot a_j\Theta_0(\nu)}.$$
Since $p_i\not\in K$ there is some $r_0>0$ such that $|y-p_i|\ge r_0$ and hence Lemma \ref{niceone:lemma} applies. It follows that for any given $\e>0$ and suitably large $n$ the following holds: 
$$\begin{aligned}
&\left| \left( \int_{\nu} \frac{\Omega_0(\nu)}{|x-p_j - \nu/n|\Theta_{n, p_j}(\nu)} \right) - \frac{1/a_j}{|x -p_j|}  \right|\\
\le& \int_{\nu} \left |\frac{\Omega_0(\nu)}{|x-p_j - \nu/n|\Theta_{n, p_j}(\nu)}  - \frac{\Omega_0(\nu)}{|x -p_j|\cdot a_j\Theta_0(\nu)} \right |< \frac{\epsilon}{\sum a_i^2}.
\end{aligned}$$
In combination with \eqref{GRF-applied} this completes the proof of the $L^\infty(K)$ convergence 
$$\theta_{\P,\alpha, n}(x)\to 1+ \frac{G}{2c^2}\sum_{i=1}^{Q}  \frac{a_i}{|x -p_i|}.$$
For sufficiently large $n$ the expression \eqref{GRF-applied} can be differentiated under the integral sign with respect to $x$; see Section \ref{LinTheory:sec}. Thus the exact same line of reasoning as above also proves the claim about the derivatives with respect to $x$.
\end{proof}

\section{One-parameter family of non-linear superposition principles}\label{GPP:sec}

The purpose of this section is to provide the proofs of our two GPP results: Theorem \ref{AngryIvaThm} and Theorem \ref{Will-Iva-alphaThm}.

\subsection{Special case of Theorem \ref{AngryIvaThm}} 
We begin by proving a special case of Theorem \ref{AngryIvaThm}, the case when $b=1$. For notational convenience we drop the explicit reference to $b=1$ in the subscript.

\begin{lemma}\label{NoahThm}
For each $\alpha\in[0,1]$ there exists a unique solution $\theta=\theta_\alpha$ of 
$$\theta^\alpha \Delta_{\geucl}\theta \dvol_{\geucl}=-4\pi \tfrac{G}{2c^2} \omega,\ \ \lim_{x\to \infty} \theta(x)=1,$$
which in addition satisfies the asymptotic conditions \eqref{asymptotic-conditions}.
The family $\theta_\alpha$ is continuous in $\alpha$ in the sense that for all convergent sequences $\alpha_n \to \alpha$ we have convergences $\theta_{\alpha_n}\to \theta_\alpha$ with all derivatives on all compact subsets of $\R^3$.
\end{lemma}

\begin{proof} 
We first establish existence and uniqueness. In the case of $\alpha=0$ there is nothing to show and so we proceed by fixing a value of $\alpha\in(0,1]$. For reasons of notational simplicity we temporarily drop $\alpha$ from the subscript. We implement  the strategy of \cite{NI} which is based on the recursive sequence 
$$ \theta_{m+1}(x):=1+\frac{G}{2c^2}\int_{\xi\in \R^3} \frac{\omega(\xi) }{|x-\xi| \theta_m^\alpha(\xi)},\ \ \theta_0(x)=1.$$
As in \cite{NI} one proves that the sequences $\theta_{2n}$ and $\theta_{2n+1}$ converge on all compacts to functions $\theta_-$ and $\theta_+$ satisfying $1\le \theta_-\le \theta_+$ and the asymptotic boundary conditions $\theta_\pm\to 1$. Furthermore, it follows that 
\begin{equation}\label{lastminutelabel}
\Delta_{\geucl}(k\theta_-+(1-k)\theta_+) =-4\pi \tfrac{G}{2c^2}\phi\cdot  \tfrac{k\theta_-^\alpha+(1-k)\theta_+^\alpha}{\theta_-^\alpha\theta_+^\alpha}
\end{equation}
for all constants $k$. 

If $\theta_-\neq \theta_+$, i.e if $\theta_-^\alpha<\theta_+^\alpha$ somewhere, then for some positive constant $k>1$ the function 
$$\theta_+^\alpha+k(\theta_-^\alpha-\theta_+^\alpha)=k\theta_-^\alpha+(1-k)\theta_+^\alpha$$
achieves a value less than $1$. By taking $k>1$ not too large, we may assume that $k\theta_-^\alpha+(1-k)\theta_+^\alpha$ reaches a positive interior minimum value. Next, we argue that $k\theta_-+(1-k)\theta_+$ for that specific value of $k$ reaches an interior minimum. There would be nothing to show if $k\theta_-+(1-k)\theta_+$ were to turn negative so we assume $k\theta_-+(1-k)\theta_+>0$ on $\R^3$. Since the function $x\mapsto x^\alpha$ is concave down, and since $k>1$, Jensen's Inequality implies 
$$(k\theta_-+(1-k)\theta_+)^\alpha\le k\theta_-^\alpha+(1-k)\theta_+^\alpha$$
over $\R^3$. It follows that the functions $(k\theta_-+(1-k)\theta_+)^\alpha$ and $k\theta_-+(1-k)\theta_+$ reach values -- and thus interior minimum values -- less than $1$. Note in addition that the function $k\theta_-+(1-k)\theta_+$ is not constant because it approaches $1$ at infinity. Since 
$$\Delta_{\geucl}(k\theta_-+(1-k)\theta_+)\le 0$$
due to \eqref{lastminutelabel}, the existence of the interior minimum value of 
the function $k\theta_-+(1-k)\theta_+$ contradicts the Strong Maximum Principle \cite{Jost}. This contradiction shows that $\theta_-=\theta_+$, and proves the existence of solutions of \eqref{TheAlphaEqn-modified} in the case of $b=1$. The uniqueness of solutions follows from the Strong Maximum Principle as in \cite{NI}, with very minor modifications to accommodate for the parameter $\alpha$.

We now focus on establishing continuity in the parameter $\alpha$. Fix $0\le \alpha<\beta\le 1$. The difference $\theta_\beta-\theta_\alpha$ satisfies 
$$\Delta_{\geucl}(\theta_\beta-\theta_\alpha)+4\pi\tfrac{G}{2c^2} \tfrac{\phi}{\theta_\alpha^\alpha\theta_\beta^\beta}\left(\theta_\alpha^\alpha-\theta_\beta^\alpha\right)=4\pi \tfrac{G}{2c^2} \tfrac{\phi}{\theta_\alpha^\alpha\theta_\beta^\beta}\left(\theta_\beta^\beta-\theta_\beta^\alpha\right),$$
which by the Mean Value Theorem and the fact that $\theta_\beta^\beta\ge \theta_\beta^\alpha\ge 1$ becomes
$$
\Delta_{\geucl}(\theta_\beta-\theta_\alpha)-4\pi \alpha \tfrac{G}{2c^2}\tfrac{\phi}{\theta_\alpha^\alpha\theta_\beta^\beta}\theta_*^{\alpha-1}\left(\theta_\beta-\theta_\alpha\right)\ge 0.
$$
It follows from the Strong Maximum Principle that $\theta_\beta-\theta_\alpha$ cannot reach a nonnegative interior maximum unless it is a constant. Given that $\theta_\beta-\theta_\alpha$ obeys a Dirichlet boundary condition, we obtain 
\begin{equation}\label{monotone}
1\le \theta_1\le \theta_\beta\le \theta_\alpha\le \theta_0.
\end{equation}

Now suppose that $\alpha_n\to \alpha$; without loss of generality we may assume that the sequence $\alpha_n$ is monotone. It follows from \eqref{monotone} that both $\theta_{\alpha_n}$ and $\frac{\phi}{\theta_{\alpha_n}^{\alpha_n}}$ are bounded in $L^2(K)$ for all compact subsets $K$. By the Interior Elliptic Regularity we see that $\theta_{\alpha_n}$ is bounded in $H^2(K)$ for all compact subsets $K$. By the Rellich Lemma and the Sobolev embedding we get a subsequential convergence of $\theta_{\alpha_n}$ to some $\theta\in C^0(K)$. However, monotonicity \eqref{monotone} ensures that the entire sequence 
$\theta_{\alpha_n}$ converges to $\theta$. The standard bootstrapping argument based on Interior Elliptic Regularity now shows that the convergence to $\theta$ happens in each and every $H^k(K)$. Taking the limit as $n\to \infty$ in the representation formula 
$$\theta_{\alpha_n}(x)=1+\frac{G}{2c^2}\int_\xi \frac{\omega(\xi)}{|x-\xi|\theta_{\alpha_n}(\xi)^{\alpha_n}}$$
shows that the limit $\theta$ solves \eqref{TheAlphaEqn-modified}. Since the said solutions are unique it must be that $\theta=\theta_\alpha$ and our proof is complete. 
\end{proof}

\subsection{Proof of Theorem \ref{AngryIvaThm}}

\begin{proof}
The uniqueness of solutions of \eqref{TheAlphaEqn-modified} follows by the same Strong Maximum Principle argument as in the proof of Lemma \ref{NoahThm}. The existence of solutions of \eqref{TheAlphaEqn-modified} in the cases when $b\neq 0$ is a consequence of Lemma \ref{NoahThm} because a function $\theta_{b,\alpha}$ serves as a solution of \eqref{TheAlphaEqn-modified} if and only if the function $\frac{\theta_{b,\alpha}}{b}$ serves as a solution of \eqref{TheAlphaEqn-modified} with $\omega$ replaced by $\frac{1}{b^{1+\alpha}}\omega$ and the asymptotic boundary condition replaced by $1$. The existence of solutions of \eqref{TheAlphaEqn-modified} for $b=0$ is established through continuity methods later on in this proof. Specifically, we show that $\lim_{b\to 0} \theta_{b,\alpha}$ exists, that it satisfies \eqref{zerocase} and that as such it defines $\theta_{0,\alpha}$.

Temporarily fix some $0<b<s\le 1$ and a value of $\alpha\in[0,1]$. Consider solutions $\theta_{b,\alpha}$ and $\theta_{s,\alpha}$ of \eqref{TheAlphaEqn-modified}. By the Mean Value Theorem the difference $\theta_{b,\alpha}-\theta_{s,\alpha}$ satisfies
$$\Delta_{\geucl}(\theta_{b,\alpha}-\theta_{s,\alpha})-\frac{4\pi\alpha\phi}{\theta_*^{\alpha+1}}\left(\theta_{b,\alpha}-\theta_{s,\alpha}\right)=0$$
for some positive function $\theta_*$. 
We see from the Strong Maximum Principle that $\theta_{b,\alpha}-\theta_{s,\alpha}$ cannot reach a nonnegative interior maximum unless it is a constant. Since $\theta_{b,\alpha}-\theta_{s,\alpha}\to b-s<0$ we arrive at $\theta_{b,\alpha}<\theta_{s,\alpha}$. This further gives $\theta_{b,\alpha}^\alpha< \theta_{s,\alpha}^\alpha$ which, when combined with the representation formula, yields  
$$\begin{aligned}
\theta_{b,\alpha}(x)=&b+\frac{G}{2c^2}\int \frac{\omega(y)}{|x-y|\theta_{b,\alpha}^\alpha(y)}\\
\ge & (b-s)+s+\frac{G}{2c^2}\int \frac{\omega(y)}{|x-y|\theta_{s,\alpha}^\alpha(y)}=(b-s)+\theta_{s,\alpha}(x).
\end{aligned}$$
Overall, we have 
\begin{equation}\label{cty-a}
0\le \theta_{s,\alpha}-\theta_{b,\alpha}\le s-b.
\end{equation}
Th estimate \eqref{cty-a} and the monotonicty formula \eqref{monotone} provide a lower bound $\theta_{b,\alpha}\ge \theta_{1,\alpha}-1\ge \theta_{1,1}-1$, valid for all $b\in (0,1]$ and all $\alpha\in[0,1]$.
For future purposes we note that  
\begin{equation}\label{lowerbound}
\theta_{b,\alpha}\big{|}_{\mathrm{supp}(\omega)}\ge \min_{\mathrm{supp}(\omega)}\frac{G}{2c^2}\int \frac{\omega(y)}{|x-y|\theta_{1,1}(y)} >0
\end{equation} 
for all $b\in (0,1]$ and all $\alpha\in[0,1]$. (Compare with Lemma \ref{WillsLemma3}.)

Next, temporarily fix $0\le \alpha<\beta\le 1$ and a value $b\in(0,1]$. Consider the solution $\theta_{\mathrm{aux}}$ of the problem 
$$\theta_{\mathrm{aux}}^\alpha \Delta_{\geucl} \theta_{\mathrm{aux}} =-4\pi \tfrac{G}{2c^2} \tfrac{1}{b^{1+\beta}}\phi,\ \ \theta_{\mathrm{aux}}\to 1.$$
Since $\alpha<\beta$ the monotonicity formula \eqref{monotone} implies 
$$\tfrac{\theta_{b,\beta}}{b}\le \theta_{\mathrm{aux}}.$$
Also note that 
$$\begin{cases}
&\left(b^{\frac{\beta-\alpha}{1+\alpha}}\theta_{\mathrm{aux}}\right)^\alpha
\Delta_{\geucl}\left(b^{\frac{\beta-\alpha}{1+\alpha}}\theta_{\mathrm{aux}}\right)=-4\pi \tfrac{G}{2c^2} \tfrac{1}{b^{1+\alpha}}\phi,\\
&b^{\frac{\beta-\alpha}{1+\alpha}}\theta_{\mathrm{aux}}\to b^{\frac{\beta-\alpha}{1+\alpha}} \text{\ \ with\ \ }b^{\frac{\beta-\alpha}{1+\alpha}}\le 1.
\end{cases}$$
The inequality \eqref{cty-a} further shows 
$$b^{\frac{\beta-\alpha}{1+\alpha}}\theta_{\mathrm{aux}}\le \tfrac{\theta_{b,\alpha}}{b}.$$ 
Since $b^{\frac{\alpha-\beta}{1+\alpha}}\le b^{\alpha-\beta}$ due to $b\in (0,1]$, we obtain $\theta_{\mathrm{aux}}\le b^{\alpha-\beta}\tfrac{\theta_{b,\alpha}}{b}$ and 
\begin{equation}\label{newmonotonicity}
\theta_{b,\beta}\le b^{\alpha-\beta}\theta_{b,\alpha},\ \ \text{i.e.}\ \ b^\beta\theta_{b,\beta}\le b^\alpha\theta_{b,\alpha}.
\end{equation}

At this stage we may repeat the argument from the end of the proof of Lemma \ref{NoahThm}, with monotonicity formula \eqref{newmonotonicity} replacing \eqref{monotone}. The conclusion is the continuity of the sequence $\theta_{b,\alpha}$ in $\alpha$ for each fixed $b\in(0,1]$. 

Finally, fix $\alpha\in[0,1]$ and consider a sequence $b_n\to b$. By \eqref{cty-a} we see that the sequence $\theta_{b_n,\alpha}$ is Cauchy in $C^0(K)$ for each compact set $K$. In fact, it is Cauchy \emph{uniformly} with respect to $\alpha$. Combining with \eqref{lowerbound} we obtain that both $\theta_{b_n,\alpha}$ and $\frac{\phi}{\theta_{b_n,\alpha}^{\alpha}}$ are uniformly Cauchy in $L^2(K)$ for each compact $K$. The Interior Elliptic Regularity and a standard bootstrapping argument show that $\theta_{b_n,\alpha}$ is uniformly Cauchy in each $H^k(K)$. The representation formula implies that the limit function $\theta$ satisfies 
\begin{equation}\label{zerocase}
\theta(x)=b+\frac{G}{2c^2}\int \frac{\omega(y)}{|x-y|\theta^\alpha(y)}.
\end{equation}

For $b>0$ the identity \eqref{zerocase} shows $\theta=\theta_{b,\alpha}$. In particular, we obtain continuity of $\theta_{b,\alpha}$ as a function of $b\in(0,1]$. Furthermore, this continuity is uniform in $\alpha$. In the case of $b=0$ we first use \eqref{zerocase} to establish the existence of solutions 
$$\theta_{0,\alpha}=\lim_{n\to \infty} \theta_{b_n,\alpha}$$
of \eqref{TheAlphaEqn-modified}. Once again, this limit is uniform in $\alpha$. Thus, the function $\theta_{b,\alpha}$ of $b\in[0,1]$ is continuous uniformly in $\alpha\in[0,1]$. The continuity of $\theta_{b,\alpha}$ as a function of $(b,\alpha)\in[0,1]\times [0,1]$ is now a consequence of the continuity of $\theta_{b,\alpha}$ as a function of $\alpha$ established earlier within this proof. 
\end{proof}

\subsection{The proof of Theorem \ref{Will-Iva-alphaThm}}
We are about to employ the strategy already used in the proof of Theorem \ref{Will-Iva-Thm1} with some slight modifications accounting for the parameter $\alpha$. Specifically, the (scaled) pullbacks we use in the proof of Theorem \ref{Will-Iva-alphaThm} are as follows:
\begin{definition}\label{pullba:defn} 
We define:
$$
\begin{cases}
&\Omega_{n,\alpha,p_j} = \Phi_{n,\alpha, p_j}\dvolg = \tfrac{1}{n^{\alpha}}\H_{n,p_j}^*\omega_{\P,\alpha,n},\\ 
&\Theta_{n,\alpha, p_j} = \tfrac{1}{n}\H _{n, p_j}^{*}\theta_{\P,\alpha, n}.
\end{cases}
$$
\end{definition}

The explicit expression for $\Omega_{n,p_j}$ now becomes
\begin{equation*}
\Omega_{n,\alpha, p_j}(x) = \sum_{i=1}^{Q}  a_i^{\alpha+1}\cdot \Omega_{0,\alpha}(x-n(p_i -p_j)).
\end{equation*}
Computation much like that employed in \eqref{BigRPP:eqn} shows that pullback of \eqref{littleGPP:eqn} under $\H_{n,p_j}$ is
$$
\Theta^{\alpha}_{n,\alpha,p_j} \Delta_{\geucl} \Theta_{n,\alpha,p_j} = -4\pi \frac{G}{2c^2}\Phi_{n,\alpha, p_j}, \  \lim_{|x| \to \infty} \Theta_{n,\alpha, p_j}(x) = \tfrac{1}{n}
$$
with $\Phi_{n,\alpha, p_j}$ and $\Theta_{n,\alpha, p_j}$ as in Definition \ref{pullba:defn}. The key to proving Theorem \ref{Will-Iva-alphaThm} is in showing that for each fixed $p_j$ the sequence of functions $\Theta_{n,\alpha, p_j}$  converges to $a_j\Theta_{0,\alpha}$ as $n\to \infty$. The latter in turn relies on the following convergence.

\begin{lemma}\label{massalpha}
Let $K$ be a compact subset of $\R^3$. If $n\ge N(K,\sigma)$ then $\Omega_{n,\alpha, p_j}=a_j^{\alpha+1} \Omega_{0,\alpha}$. In particular, we have 
$$\Omega_{n,\alpha, p_j} \to a_j^{\alpha +1} \Omega_{0,\alpha} \text{\ \ as\ \ } n\to \infty.$$
This convergence is uniform with all derivatives over all compact subsets of $\R^3$.
\end{lemma}

The inequality 
\begin{equation*}
\begin{aligned}
\Delta_{\geucl}(\Theta_{n,\alpha, p_j}^{\alpha +1}) =&(\alpha+1)\Theta^{\alpha}_{n, \alpha, p_j} \Delta_{\geucl} \Theta_{n,\alpha, p_j} + (\alpha+1) |d\Theta _{n,\alpha, p_j}|^2\\
\geq & -4\pi (\alpha+1) \frac{G}{2c^2}\Phi_{n, \alpha,p_j}.
\end{aligned}
\end{equation*}
when used in place of \eqref{keyinequality} and \eqref{BigDeal} gives us an upper bound on $\Theta_{n,\alpha, p_j}^{\alpha +1}$ and the following results.

\begin{lemma}\label{ThetaBounds:propa}\ 
\begin{enumerate}
\item The functions $\Theta_{n,\alpha, p_j}$ are bounded in $L^\infty(\R^3)$, independently of $n$ and our choice of $p_j$:
$$\|\Theta_{n,\alpha, p_j}\|_{L^\infty(\R^3)}\le M_+ \text{\ \ for all\ \ } n, p_j.$$
\medbreak
\item Given a compact set $K$ there exists a positive constant $M_-=M_-(K)$  so that 
\begin{equation*}
\Theta_{n, \alpha,p_j}(y) \ge a_j^{1+\alpha}\cdot M_{-}(K)
\end{equation*}
for all $n, p_j $ and all $y\in K$.
\medbreak
\item\label{part3} Let $\e>0$, and let $p_j$ be fixed. There exists $L>1$ such that for all $|y| > L$ there exists $N \in \mathbb{N}$ so that when $n\ge N$ we have
$\Theta_{n,\alpha, p_j}(y) < \e$.
\end{enumerate}
\end{lemma}

The steps outlined in Lemma \ref{ExUnT:prop} can be executed in the new framework as well, leading us to the following convergence. 

\begin{lemma}\label{ExUnT:propa}
For a fixed $p_j$ the sequence $\Theta_{n, \alpha,p_j}$ converges uniformly with all derivatives over all compact subsets of $\R^3$ to $a_j\Theta_{0,\alpha}$.
\end{lemma}

\begin{proof}
Upper bounds on $\norm{\Phi_{n, p_j } \Theta_{n, \alpha,p_j}^{-\alpha}}_{L^2}$ and $\norm{\Theta_{n,\alpha,p_j}}_{L^2}$ are obtained by applying estimates from Lemma \ref{ThetaBounds:propa}. In combination with Elliptic Regularity Estimates
$$
\norm{\Theta_{n,\alpha,p_j}}_{H^2(K')} \cle \norm{\Phi_{n, p_j } \Theta_{n, \alpha,p_j}^{-\alpha}}_{L^2(K)} + \norm{\Theta_{n,\alpha,p_j}}_{L^2(K)}
$$
and Rellich Lemma these upper bounds lead to a diagonal sequence $\Theta_{n,\alpha, p_j}^{(n)}$ which converges (uniformly with two derivatives over all compact subsets of $\R ^3$) to some function $\Theta_\infty$. By virtue of part \eqref{part3} of Lemma \ref{ThetaBounds:propa} we know that 
$ \Theta_{\infty}(y) =0$ as $|y|\to \infty$. In combination with Lemma \ref{massalpha} we further see that $\Theta_\infty$ solves the GPP
\begin{equation*}
\Theta_{\infty}^{\alpha} \Delta_{\geucl} \Theta_{\infty} \dvol_{\geucl} = -4\pi \tfrac{G}{2c^2} a_j^{\alpha+1}\Omega_{0,\alpha},  \ \lim_{|x| \to \infty} \Theta_{\infty}(x) =0.
\end{equation*} 
Theorem \ref{AngryIvaThm} states that solutions to the above are unique, i.e that $\Theta_\infty=a_j \Theta_{0,\alpha}$. The remainder of the proof proceeds exactly as in Lemma \ref{ThetaBounds:propa}. 
\end{proof}

To complete the proof of Theorem \ref{Will-Iva-alphaThm} we employ the Green's representation formula:
$$\begin{aligned}
\theta_{\P,\alpha,n}(x)=&1+\frac{G}{2c^2}\int_{\xi} \frac{\omega_{\P,\alpha,n}(\xi)}{|x-\xi|\theta^{\alpha}_{\P,\alpha,n}(\xi)}\\ 
=&1+\frac{G}{2c^2}\sum_{i=1}^{Q} a_i^{\alpha+1}\int_{\xi} \frac{(\H_{n,p_i})_{*}\Omega_{0,\alpha}(\xi)}{|x-\xi|\cdot \frac{1}{n^{\alpha}}\theta^{\alpha}_{\P,\alpha,n}(\xi)}\\
=& 1+ \frac{G}{2c^2}\sum_{i=1}^{Q} a_i^{\alpha+1}\int_{\nu} \frac{\Omega_{0,\alpha}(\nu)}{|x-p_i - \nu/n|\,\Theta^{\alpha}_{n, \alpha,p_i}(\nu)}.
\end{aligned}
$$
As a corollary of Lemma \ref{niceone:lemma} and Lemma \ref{ExUnT:propa} we now have 
$$
\left\|\partial_x^l\left(\frac{\Omega_{0,\alpha}(\nu)}{|x - \nu/n|\Theta_{n, p_j,\alpha}^\alpha(\nu)}\right) - \partial_x^l\left(\frac{\Omega_{0,\alpha}(\nu)}{|x|\cdot a_j^\alpha\,\Theta_{0,\alpha}^\alpha(\nu)}\right)\right\|_{L^\infty}\to 0
$$
as $n\to \infty$. Note that for any fixed $r_0>0$ the convergences are uniform with respect to $x\in \R^3\smallsetminus B(0,r_0)$ and that
$$\int_{\nu} \frac{\Omega_{0,\alpha}(\nu)}{|x-p_i - \nu/n|\Theta^{\alpha}_{0, \alpha}(\nu)}=1$$
due to $\alpha$-normalization (see Definition \ref{newnormalization:definition}). 
Overall, we obtain the convergence of $\theta_{\P,\alpha,n}$, uniform with all derivatives on all compact subsets of $\R^3\smallsetminus\{p_1, .., p_Q\}$, towards 
$$1+ \frac{G}{2c^2}\sum_{i=1}^{Q} \frac{a_i}{|x-p_i|}.$$

\section{Concluding remarks}
To complete the program of representing a relativistic cloud of matter as a cumulative effect of point-sources we in addition need to execute the following.

\begin{itemize}
\item\emph{Discretization of sources}. For a given matter distribution $\omega$ there needs to be way of associating a suitable sequence $\P_Q=\{ (p_1, a_1), ... , (p_{Q}, a_{Q})\}$ indexed by $Q$ so that in some kind of limit as $Q\to \infty$ one recovers $\omega$. Note that one does not expect $a_i$ to simply be $\omega\big{|}_{p_i}$ because of non-linear interaction effects. Instead, since the conformal factor $\theta$ mimics the gravitational potential it is expected that the parameters $a_i$ are related to $\Delta_{\geucl}\theta\big{|}_{p_i}$. In other words, it is expected that the computation of parameters $a_i$ relies on the RPP. However, it would be far more optimal to have an iterative algebraic (and in particular non-PDE-based!) algorithm which determines $a_i$ based on the values of $\omega\big{|}_{p_j}$ for various $j$. We are in the process of developing such an algorithm, along with a theorem which quantifies the extent to which the parameters $a_i$ approximate $\Delta_{\geucl}\theta\big{|}_{p_i}$.
\medbreak
\item \emph{Employment of limits}. Ideally, for a given matter distribution $\omega$ and an approximating sequence $\P_Q$ discussed above one would also have a theorem along the following lines: the limit as $Q\to \infty$ of superpositions associated with $\P_Q$ in the sense of our Definition \ref{dust:defn} is the matter $\omega$ being discretized in the first place. A paper \cite{TI} is being written on this subject; it employs the concept of intrinsic flat limits \cite{SW-JDG}.
\end{itemize}

\bibliographystyle{plain}

\end{document}